\documentclass[12pt,letterpaper]{article}

\usepackage[margin=1in]{geometry}
\usepackage[T1]{fontenc}
\usepackage[utf8]{inputenc}
\usepackage{mathptmx}

\usepackage{amsmath,amssymb,amsthm}
\usepackage{mathtools}
\usepackage{bm}

\usepackage{booktabs}
\usepackage{graphicx}
\usepackage{float}
\usepackage{caption}
\usepackage{natbib}
\setcitestyle{aysep={}}

\usepackage{setspace}
\usepackage{lineno}
\usepackage{comment}
\usepackage[toc,page]{appendix}

\usepackage{booktabs}
\usepackage{threeparttable}
\usepackage{siunitx}

\usepackage{enumitem}
\usepackage[hidelinks]{hyperref}
\hypersetup{
    colorlinks=true,
    citecolor=blue,
    linkcolor=blue,
    urlcolor=blue  
}

\hypersetup{pdftitle={Behavioral CES},pdfauthor={Tingmingke Lu and Zhenyi Wang}}
\usepackage[protrusion=true,expansion=false]{microtype}
\newtheorem{proposition}{Proposition}
\newtheorem{lemma}{Lemma}

\newtheorem{assumption}{Assumption}

\newcommand{\dln}{\Delta\ln}
\DeclareMathOperator{\Cov}{Cov}
\DeclareMathOperator{\Var}{Var}
\DeclareMathOperator*{\plim}{plim}
\newcommand{\R}{\mathbb{R}}

\title{Bounded Attention and Attenuated Elasticities}

\author{Tingmingke Lu\thanks{WU (Vienna University of Economics and Business). Email: tingmingke.lu@wu.ac.at.} \and Zhenyi Wang\thanks{WU (Vienna University of Economics and Business). Email: zhenyi.wang@wu.ac.at.}}

\date{\today}

\begin{document}

\maketitle

\begin{abstract}
\noindent

We study how bounded attention affects the structural estimation of the elasticity of substitution. In a sparse-max model, equilibrium prices and expenditure shares are observationally equivalent to those of a rational market. Its elasticity equals the structural elasticity scaled by the attention weight. The supply-side orthogonality condition of the covariance-based estimator holds at this attenuated elasticity but fails at the truth. Absent outside information, the estimator converges to the attenuated elasticity. Bounded attention leads to an identification failure that finite-sample and weak-instrument corrections cannot resolve. Because attention cannot exceed one, the estimate provides a lower bound on the structural elasticity.


\end{abstract}

\newpage
\section{Introduction}
\label{sec:intro}

Governments have begun to regulate how prices are perceived, not only what they are. Since July 2024, France has required supermarkets larger than 400 square meters to flag, for two months, any product whose package shrank without a corresponding reduction in its shelf price, on the grounds that shoppers typically fail to notice the change at the time of purchase.\footnote{Order of April 16, 2024, amended June 28, 2024, enforced by the Directorate General for Competition, Consumer Affairs and Fraud Control. For more details, see \href{https://entreprendre.service-public.gouv.fr/actualites/A17378}{https://entreprendre.service-public.gouv.fr/actualites/A17378.}} The economic rationale for such rules is that demand responds to the price consumers perceive, and perceived prices need not coincide with actual prices \citep{chetty2009salience, strulov2023more, d2021exposure}. Yet the purchases made by inattentive consumers are the raw input for structural demand estimation.

The workhorse constant elasticity of substitution (CES) demand system in quantitative trade and spatial economics uses those purchases to recover structural parameters. The elasticity at its core governs the measured gains from trade \citep{arkolakis2012new, costinot2014trade}, the welfare value of new varieties \citep{feenstra1994new, broda2006globalization}, the pass-through of tariffs \citep{amiti2019impact, fajgelbaum2020return}, and the measurement of aggregated cost of living \citep{redding2020measuring, hoste2025uncovering}. Standard estimation recovers this elasticity from prices and expenditure shares by treating the observed shares as rational CES responses to the listed prices, under the assumption that consumers fully attend to all the prices they face. But if consumers do not allocate full attention to all prices, what does the standard estimator recover? This paper investigates.

We embed CES demand in the sparse-max model of bounded attention \citep{gabaix2014sparsity}, where consumers respond to part of each price change and anchor the rest of their perceptions to a default price while facing a real budget constraint. From the perspective of an econometrician assuming full attention, we then employ the latest refinement of the \citet{feenstra1994new} covariance restrictions-based estimator by \citet{soderbery2015estimating} to ask what structural information can be recovered from market data.

We show that the equilibrium prices and expenditure shares of an inattentive market are observationally equivalent to those of a rational market whose elasticity is the structural one scaled by the attention weight. The equivalence survives structural estimation. The orthogonality condition that underpins the supply-side identification strategy holds for the attenuated elasticity but fails for the structural one. The estimator therefore returns the attenuated elasticity, which is a pseudo-true parameter in the sense of \citet{andrews2026true}. We derive these results in closed form for a representative consumer and confirm them in simulated markets of heterogeneous consumers.

Our results rely on the budget adjustment \citet{gabaix2014sparsity} builds into sparse max. The consumer equates marginal rates of substitution to perceived prices, but her spending at true prices exhausts the budget. Due to this asymmetry, demand responds to true prices with the structural elasticity scaled by the attention weight. The equivalent economy is itself a rational market, so the standard model is correctly specified for the data generated by attention friction. Every diagnostic is a function of prices and shares, and both markets yield the same distribution for any such function. Consequently, each test passes at the attenuated value, and nothing in the joint distribution of prices and shares reveals the friction. Because attention cannot exceed one, the estimated elasticity is a lower bound on the structural elasticity. Therefore, restrictions strictly weaker than the full attention that point identification imposes still yield informative bounds on the true elasticity.

This paper adapts the sparse-max framework of \citet{gabaix2014sparsity} to the log-linear CES demand system and addresses the structural econometrics of estimating the elasticity of substitution from market data. The covariance-based estimator studied in this paper was developed by \citet{feenstra1994new} and refined by \citet{broda2006globalization} and \citet{soderbery2015estimating}. \citet{mackay2025estimating} recently generalized the covariance-restriction principle to broader models of supply and demand under imperfect competition. We show that if the underlying consumer behavior that generates the population moments of the estimator is driven by bounded rather than full attention, the estimator systematically misidentifies the parameter.

Moreover, we show that the divergence between the true elasticity and the returned value persists in the population limit if consumers do not allocate full attention to all prices. The observational equivalence behind that persistence has a macroeconomic precedent in \citet{angeletos2021myopia}, who recast incomplete information as two behavioral distortions, myopia toward the future and anchoring to the past. We prove a static counterpart of that equivalence for structural demand estimation. Their distortions grow with the strength of the general equilibrium feedback. Our result runs through the supply-side identifying restriction, which holds at the attenuated value and fails at the truth, so the estimator treats bounded attention as less elastic demand. Our paper also connects to the tax salience literature, particularly \citet{chetty2009salience} and \citet{kroft2024salience}, who show how consumer inattention scales demand responses linearly. We contribute to this literature by linking this behavioral scaling to the asymptotic failure of covariance-based demand estimators.

\section{Analysis}
\subsection{The Sparse-Max CES System}
\label{sec:setup}

We work within a single product nest. The nest contains varieties indexed by \(v\), with prices \(p_v > 0\), true taste weights \(\phi_v > 0\), and a common elasticity of substitution \(\sigma^r > 1\), where the superscript \(r\) marks it as the structural parameter under rationality.\footnote{We suppress the nest subscript because every result holds nest by nest.} The consumer has CES preferences. She chooses quantities \(c_v\) to maximize \( U = \big(\sum_v \phi_v^{1/\sigma^r} c_v^{(\sigma^r-1)/\sigma^r}\big)^{\sigma^r/(\sigma^r-1)} \) subject to \( \sum_v p_v c_v = E \), in which \(E\) is the nest expenditure.

\paragraph{The rational benchmark.}
Under full attention, the Marshallian expenditure share of variety \(v\) is
\begin{equation}\label{eq:rational_share}
  s_v \;=\; \frac{\phi_v\, p_v^{\,1-\sigma^r}}{\sum_j \phi_j\, p_j^{\,1-\sigma^r}},
\end{equation}
which adopts the multiplicative taste convention of \citet{feenstra1994new} so that \(\phi_v\) enters the share directly.

To investigate how a consumer with bounded attention responds to price fluctuations when anchoring perceptions on a default baseline, we embed CES demand within a sparse-max framework. This setting allows the consumer to maintain a sparse perception of the price vector. As in \citet{gabaix2014sparsity}, the consumer anchors on a default price and adjusts partially toward the truth. We specify the adjustment in logs,
\begin{equation}\label{eq:perceived}
  \ln p_v^s \;=\; m\,\ln p_v + (1-m)\,\ln \bar{p},
  \qquad m \in (0,1],
\end{equation}
in which \(p_v^s\) is the perceived price with the superscript \(s\) marking it as a parameter under sparsity, \(\bar{p} > 0\) is a nest-level anchor price, and \(m\) is the attention weight. Full attention at \(m=1\) yields \(p_v^s = p_v\) and restores rationality, while lower \(m\) pulls perceived prices toward the anchor. We rule out \(m = 0\), under which every perceived price equals the anchor and demand does not respond to true prices.

In our model, a consumer who forms a sparse perception of the log price applies the attention model to the variable on which the demand system operates as CES demand and the consumer's first-order conditions are linear in log prices. We consider \(\bar{p}\) as fixed within the estimation window and treat it as the sparse consumer's default price in the sense of \citet{gabaix2014sparsity}, not as a statistical forecast, so that all variation in perceived prices comes from true prices and attention.

The perception ratio
\begin{equation}\label{eq:R}
  R_v \;\equiv\; \frac{p_v^s}{p_v} \;=\; \left(\frac{\bar{p}}{p_v}\right)^{1-m}
\end{equation}
exceeds one for varieties priced below the anchor and falls short of one for those priced above it, so a consumer who pays little attention overstates the price of a cheap variety and understates the price of an expensive one. In the sparse-max framework, \(m\) is not free but is chosen to balance the utility cost of misperceiving a price against a cognitive penalty for attention, and the optimal weight rises with the price volatility of the variety, its expenditure share, and the elasticity \citep{gabaix2014sparsity}.\footnote{We treat \(m\) as a structural parameter and ask what the data reveal about it. We do not use the optimality calculation in any identification result below. We return to its optimal value in the discussion.}

\paragraph{Sparse demand and the budget adjustment.}
Following \citet{gabaix2014sparsity}, we adopt the sparse max operator with budget adjustment, which allows a sparse consumer to set marginal rates of substitution equal to perceived price ratios while exhausting her true budget. The operator solves the rational optimization at perceived prices and then rescales the budget so that true spending equals \(E\). In particular, quantities at perceived prices are \(c_v^s = A\,\phi_v\,(p_v^s)^{-\sigma^r}\), in which \(A\) is a common scale. The true budget $\sum_j p_j c^s_j = E$ pins $A$.

Carrying out the rescaling, the perceived-price aggregator cancels and the sparse expenditure share is
\begin{equation}\label{eq:sparse_share}
  s_v^s \;=\; \frac{p_v\,\phi_v\,(p_v^s)^{-\sigma^r}}{W},
  \qquad
  W \;\equiv\; \sum_j p_j\,\phi_j\,(p_j^s)^{-\sigma^r}.
\end{equation}
The denominator exposes the asymmetry of sparse demand, in which perceived prices, through \((p_j^s)^{-\sigma^r}\), govern substitution across varieties, while true prices, through \(p_j\), enforce budget balance.

\paragraph{Observational equivalence.}
Define the effective taste weight
\begin{equation}\label{eq:eff_taste}
  \tilde\phi_v \;\equiv\; \phi_v\left(\frac{p_v^s}{p_v}\right)^{-\sigma^r} \;=\; \phi_v\,R_v^{-\sigma^r}.
\end{equation}
Substituting \eqref{eq:eff_taste} into the rational share \eqref{eq:rational_share}, with \(\tilde{\phi}\) in place of \(\phi\), reproduces the sparse share, which allows us to state the following result. (All proofs are relegated to Appendix \ref{app:proofs}.)

\begin{lemma}[Observational equivalence]\label{lem:obs_equiv}
The bounded-attention expenditure shares \eqref{eq:sparse_share} are identical to rational CES shares \eqref{eq:rational_share} in which the true taste weight \(\phi_v\) is replaced by the effective taste weight \(\tilde\phi_v = \phi_v R_v^{-\sigma^r}\).
An econometrician who observes \((s_v^s, p_v)\) and maintains rational CES recovers \(\tilde\phi_v\), not \(\phi_v\).
\end{lemma}

The equivalence here carries its classical meaning, as in \citet{koopmans1949identification} and \citet{rothenberg1971identification}. There exists a rational economy that generates the same observables, so no procedure that uses only those observables can tell the two apart. The wedge \(R_v^{-\sigma^r}\) is the entire content of bounded attention as seen from the data. It is not a measurement error. It is a deterministic function of the price, and through \eqref{eq:R}, it moves whenever prices move.

Because of this behavioral wedge, \(R_v^{-\sigma^r}\), working on the demand side alone does not recover the true elasticity of substitution \(\sigma_r\). To show that, we write \(\dln \) for the change in the log over time and denote \( \dln R_{vt} \coloneqq \ln R_{v,t} - \ln R_{v,t-1}\). We impose a fixed anchor and a homogeneous attention within the nest, so the perception ratio satisfies \[\dln R_{vt}  = -(1-m)(\ln p_{v,t} - \ln p_{v,t-1}) = -(1-m)\dln p_{vt}. \] Taking logs of the observationally equivalent share, first-differencing, and absorbing the nest-time aggregate into a period effect \(\eta_t\) gives the demand equation
\begin{equation}\label{eq:demand_exact}
  \dln s_{vt}^s
  \;=\; \eta_t + \dln\phi_{vt} - \sigma^r\,\dln R_{vt} + (1-\sigma^r)\,\dln p_{vt}
  \;=\; \eta_t + \dln\phi_{vt} + (1-\sigma^r m)\,\dln p_{vt}.
\end{equation}

Equation \eqref{eq:demand_exact} is an exact identity that is linear in \(\dln p_{vt}\) because the log anchor makes \(\ln R_v\) linear in \(\ln p_v\).\footnote{As long as the anchor is common to all the varieties in the nest, its change enters \(\dln R_{vt}\) in the same way for every variety. Therefore, it is absorbed by the period effect \(\eta_t\) and cancels out from all changes taken relative to a reference variety.} The price coefficient is \(1-\sigma^r m\), and the demand residual is the true taste shock \(\dln\phi_{vt}\), which carries a unit coefficient.\footnote{Note that attention does not contaminate the demand equation residual.} An econometrician who estimates this equation, by any method that handles the simultaneity between prices and shares, recovers the coefficient \(1-\sigma^r m\) and reports the apparent elasticity
\begin{equation}\label{eq:apparent}
  \hat{\sigma} \;=\; \sigma^r m,
\end{equation}
which is the rational structural elasticity attenuated by the attention weight \(m\). Equation \eqref{eq:apparent} indicates that the behavioral demand response equals the rational response scaled by the attention the consumer places on the price. The log perception rule makes this scaling exact, so the scaling factor equals the attention weight at every price and not only in a neighborhood of the anchor.\footnote{The log rule and the additive rule \(p_v^s = m p_v + (1-m)\bar p\) agree to first order at the anchor, where both give pass-through \(m\) and local elasticity \(\sigma^r m\). Away from the anchor, the additive rule passes through a share \(m p_v / (m p_v + (1-m)\bar p)\) of each log price change, and this share rises with the price. The local elasticity is \(\sigma^r\) times a share that varies across varieties. Under the additive rule, the estimator returns a value strictly below the structural elasticity, now a covariance-weighted average of these local elasticities rather than the constant \(\sigma^r m\). Appendix \ref{app:level} contains the full characterization.}

Equation \eqref{eq:apparent} is one equation in two unknowns. Any preference-attention pair on the rectangular hyperbola
\begin{equation}\label{eq:identset}
  \mathcal{M} \;=\; \big\{(\sigma^r, m) \in \R_{>1}\times(0,1] : \sigma^r m = \hat{\sigma}\big\}
\end{equation}
generates the same apparent elasticity. Therefore, the structural elasticity is identified with the identified set \(\mathcal{M}\).

Accounting for the joint equilibrium of supply and demand does not remove the friction. Under the standard identifying restrictions, an econometrician who estimates the elasticity from prices and shares still converges to the attenuated elasticity in \eqref{eq:apparent}, not the structural elasticity. The next section proves this.

\subsection{Identification}
\label{sec:identification}

We complete our model by specifying a structural supply side, mapping our equilibrium framework directly to the setup of \citet{feenstra1994new}. This structural alignment ensures that our behavioral framework satisfies the identifying assumptions required to implement his identification strategy and empirical estimator.

\paragraph{The market.}
Each variety is produced by a monopolistically competitive firm that prices at a constant markup over a marginal cost that rises with output,
\begin{equation}\label{eq:supply}
  p_{vt} \;=\; \mu\,\nu_{vt}\,x_{vt}^{\,\omega},
\end{equation}
where \(x_{vt}\) is output, \(\nu_{vt} > 0\) is a multiplicative cost shock, \(\omega \ge 0\) is the inverse supply elasticity, and \(\mu\) is the markup.\footnote{A firm that prices against the perceived elasticity charges a different markup than one that prices against the rational elasticity, but with \(\sigma^r\) and \(m\) being constant, either markup is a constant, and any constant cancels from the differenced equations below. Therefore, the estimating moments are the same in the two cases. Nothing below depends on the choice.}

Following \citet{feenstra1994new}, we remove the period effect \(\eta_t\) in \eqref{eq:demand_exact} by differencing each variety against a reference variety \(k\) in the same nest.\footnote{From here on, every \(\dln\) represents the double difference over time and relative to the reference variety.} Doing so yields 
\begin{equation}\label{eq:demand_dbdiff}
  \dln s_{vt}^s
  \;=\; \dln\phi_{vt} + (1-\sigma^r m)\,\dln p_{vt}.
\end{equation}
Next, solving the demand equation \eqref{eq:demand_dbdiff} together with the pricing equation \eqref{eq:supply} yields the equilibrium price change
\begin{equation}\label{eq:reduced}
  \dln p_{vt} \;=\; \frac{\omega\,\dln\phi_{vt} + \dln\nu_{vt}}{D},
  \qquad
  D \;\equiv\; 1 + \omega\,\sigma^r m \;>\; 0.
\end{equation}
The reduced-form relation in Equation \eqref{eq:reduced} establishes that equilibrium prices respond to both taste and cost changes. This joint determination implies non-zero covariances between the log price variation \(\dln p_{vt}\) and the two structural errors, the demand residual \(\dln\phi_{vt}\) and the supply-side cost shock \(\dln\nu_{vt}\).

\paragraph{The identifying moment.} The \citet{feenstra1994new} strategy identifies the elasticity as the value at which the demand residual implied by rational CES is orthogonal to the cost shock. 

\begin{assumption}[Independent structural shocks]\label{ass:orth}
The structural taste and cost shocks are uncorrelated, \(\Cov(\dln\phi_{vt}, \dln\nu_{vt}) = 0\).
\end{assumption}

The \citet{feenstra1994new} estimator chooses \(\sigma^r\) to make the implied taste residual orthogonal to the cost side. For a candidate elasticity \(\sigma\), the estimator inverts the rational demand equation in relative terms to obtain the taste residual
\begin{equation}\label{eq:recovered_resid}
  \dln\tilde\phi_{vt}(\sigma) \;=\; \dln s_{vt} + (\sigma-1)\,\dln p_{vt},
\end{equation}
with all changes taken relative to a reference variety \(k\) in the nest. At the true \(\sigma^r\), this residual is the true taste shock \(\dln\phi_{vt}\). The estimator then checks whether \(\dln\tilde{\phi}_{vt}(\sigma)\) is uncorrelated with the supply shock.

Within the sparse-max CES framework, however, we substitute the empirical market share changes \(\dln s_{vt}^s\), which are the observed data analogs to the latent shares \(\dln s_{vt}\), along with the double-differenced demand equation \eqref{eq:demand_dbdiff} and the reduced-form equilibrium price relation \eqref{eq:reduced} into the system. Invoking the structural orthogonality restrictions in Assumption \ref{ass:orth} then yields the analytical, closed-form moment condition expression that holds variety by variety,
\begin{equation}\label{eq:moment}
  \Cov\!\big(\dln\tilde\phi_{vt}(\sigma),\, \dln\nu_{vt}\big)
  \;=\; (\sigma - \sigma^r m)\,\frac{\Var(\dln\nu_{vt})}{D}.
\end{equation}

When evaluated at the attenuated value (i.e., \(\sigma = \sigma^r m\)), the moment vanishes. It is strictly positive at the true elasticity and equal to \(\sigma^r(1-m)\Var(\dln\nu_{vt})/D\) whenever \(m<1\). The orthogonality condition that the \citet{feenstra1994new} estimator drives to zero is therefore satisfied at the attenuated value and violated at the truth.\footnote{At \(\sigma^r m\), the implied residual in \eqref{eq:recovered_resid} reduces to the true taste shock, which is independent of the cost shock. At \(\sigma^r\), it equals the effective-taste change of Lemma \ref{lem:obs_equiv}, because \(\dln\tilde\phi_{vt}(\sigma^r) = \dln\phi_{vt} + \sigma^r(1-m)\dln p_{vt} = \dln\tilde\phi_{vt}\), which is why the notation reuses the tilde.}

Two requirements stand between these moment-level facts and an estimator-level result. The first is that the equivalent market must lie in the region the estimator searches. For attention low enough that \(\sigma^r m \le 1\), the equivalent market leaves the relevant CES region, the inversion from the estimated slopes to the elasticity has no root above one, and the estimator has no consistent value to find, so the attenuated product must exceed one. 

The second concerns the sampling frame. The identified set \(\mathcal{M}\) in expression \eqref{eq:identset} is a property of the population distribution of \(\{\dln p_{vt}, \dln s_{vt}^s\}\), so inconsistency is a statement about unlimited data, and stating it requires the estimator's own asymptotic frame formed by a long panel over a fixed set of varieties. Along such a long panel, the finite-sample and weak-instrument biases of a correctly specified estimator vanish, while the moment violation at the truth, \(\sigma^r(1-m)\Var(\dln\nu_{vt})/D\), is a population quantity that a long panel cannot fix. 

We collect these requirements, Feenstra's rank condition, and an explicit consistency condition in the following assumption.

\begin{assumption}[Regularity for the estimator]\label{ass:reg}
(i) \(\omega > 0\) and \(\sigma^r m > 1\), so the equivalent rational market lies in the relevant CES region. (ii) The ratio of the relative taste-shock variance to the relative cost-shock variance differs across at least two non-reference varieties, so the nest contains at least three varieties, which is Feenstra's rank condition. (iii) The panel grows long, \(T \to \infty\), with the variety set fixed. (iv) In a rational CES market that satisfies (i) to (iii) and is correctly specified, the Feenstra and Broda-Weinstein estimator, in its product-moment and limited-information maximum-likelihood forms, is consistent for the elasticity and the inverse supply elasticity.
\end{assumption}

We are now in a position to state our central result.

\begin{proposition}[Complete observational equivalence]\label{prop:main}
Under log perception with a fixed anchor and Assumption \ref{ass:orth}, the equilibrium data \(\{\dln p_{vt}, \dln s_{vt}^s\}\) are observationally equivalent to those of a fully rational market with elasticity \(\sigma^r m\), supply elasticity \(\omega\), and taste shock \(\dln\phi_{vt}\).
Consequently,
\begin{enumerate}[label=(\roman*), leftmargin=2.2em, itemsep=2pt]
  \item the identifying moment \eqref{eq:moment} equals zero at \(\sigma=\sigma^r m\) and is strictly positive, equal to \(\sigma^r(1-m)\Var(\dln\nu_{vt})/D\), at the true elasticity \(\sigma=\sigma^r\);
  \item under Assumption \ref{ass:reg}, the Feenstra and Broda-Weinstein estimator converges to \(\hat{\sigma} = \sigma^r m\) and \(\hat\omega = \omega\).
\end{enumerate}
The supply-side strategy returns the attenuated elasticity and does not recover \(\sigma^r\).
\end{proposition}

Appendix \ref{app:proofs} proves this by constructing an explicit rational economy that reproduces the joint distribution of prices and shares period by period, so the estimator, a function of that distribution, returns the parameters of the rational economy.

While recent refinements to Feenstra's framework explicitly target finite-sample fragility and weak-instrument issues \citep[e.g.,][]{soderbery2015estimating,andrews2019weak}, Proposition \ref{prop:main} demonstrates that behavioral attenuation operates through an entirely distinct channel. Whenever the attention weight deviates from unity, this attenuation emerges as a property of the population moment rather than an attribute of the sampling distribution. Because the population estimand is strictly pinned down as \(\sigma^r m\), the structural wedge \(\sigma^r(1-m)\) persists irrespective of panel length or first-stage instrument strength. Consequently, this bias cannot be mitigated by collecting more periods or securing stronger instruments. The discrepancy is a fundamental identification failure driven by model misspecification.

Equation \eqref{eq:moment} is the population value of the covariance restriction that vanishes at \(\sigma^r m\). An estimator that drives the sample covariance of the recovered demand and supply residuals to zero, therefore, converges to the same attenuated product.

Because point identification of \(\sigma^r\) relies exclusively on the maintained assumption of full attention, the point estimates are only as credible as that assumption itself. Evaluating policy by feeding the \(\hat{\sigma}\) estimate into standard welfare formulas yields distorted measures. Resolving this distortion requires isolating \(\sigma^r\) from the composite term in Equation \eqref{eq:apparent}, a decomposition that observational market data alone cannot achieve.

\paragraph{What the estimate still delivers.} The set \(\mathcal{M}\) characterizes the informative bounds of the model. Due to the observational equivalence formalized in Proposition \ref{prop:main}, every parameter vector \((\sigma^r, m) \in \mathcal{M}\) induces the same joint distribution of prices and shares as a fully attentive market operating under elasticity \(\hat{\sigma}\). This mutual indistinguishability implies that the identified set for \((\sigma^r, m)\) is the intersection of \(\mathcal{M}\) with the parameter support dictated by Assumption \ref{ass:reg}(i).

Because the attention parameter is naturally bounded above by unity ($m \le 1$), the recovered apparent elasticity establishes a lower bound for the true parameter, dictating that \(\sigma^r \ge \hat{\sigma}\). Any economically defensible lower threshold on attention, \(m \ge m_0\), conversely caps the identified set from above, mapping out the structural interval \(\sigma^r \in [\hat{\sigma},\, \hat{\sigma}/m_0]\). Empirical candidates for this attention floor can be directly informed by existing benchmarks in the behavioral and salience literature \citep[e.g.,][]{chetty2009salience, dellavigna2009psychology}. This bounding strategy that constructs set identification from behavioral restrictions strictly weaker than those required for dogmatic point identification falls directly within the classical sensitivity and partial-identification traditions of \citet{leamer1981demand}, \citet{manski2003partial}, and \citet{mackay2025estimating}. In the next section, we confirm the analytical results by running the actual estimator on simulated markets.

\section{Quantitative Illustration}
\label{sec:quant}

In our data generating process, the attention specification follows the sparse-attention formulation of \citet{gabaix2014sparsity}, applied through the anchored log perception rule of Equation \eqref{eq:perceived}. The structural elasticity is \(\sigma^r = 15\) and the focal attention weight is \(m = 0.5\), which makes the attenuated target at \(\sigma^r m = 7.5\). Therefore, the attenuated target is far enough from the truth that we won't mistake the gap for noise and the requirement \(\sigma^r m > 1\) of Assumption \ref{ass:reg}(i) holds in each simulation.

We apply the Feenstra product-moment regression with the limited-information maximum likelihood and Fuller refinements that \citet{soderbery2015estimating} recommends. The inverse supply elasticity \(\omega = 1.5\) gives an upward-sloping supply curve. The shock dispersions are drawn uniformly per variety from the given ranges, so the taste-to-cost variance ratio differs across varieties and the rank condition of Assumption \ref{ass:reg}(ii) holds by construction with fifteen varieties against the required three. At these dispersions, the panels produce first-stage strength in the single digits at conventional panel lengths and above the rule-of-thumb threshold only for the longest panels.

We design replications that build balanced panels in five steps. First, we draw one standard deviation for the taste shock and one for the cost shock of each variety from the given ranges, and both stay fixed throughout. Second, we draw the taste shocks \(\dln\phi_{vt}\) and the cost shocks \(\dln\nu_{vt}\) as normal draws with mean zero and those standard deviations, independently of each other, across varieties, and over time.
Independent draws leave the two shocks uncorrelated, so Assumption \ref{ass:orth} holds by construction. Then, we form the equilibrium price changes from the reduced form \eqref{eq:reduced}.

Next, because the shocks are formed in changes while expenditure shares are statements about levels, we turn changes into levels. We let baseline log prices form an evenly spaced grid across \(\ln\bar p \pm 0.35\), so varieties start on both sides of the anchor, and we let baseline tastes equal one. Each later period adds the accumulated price changes and the accumulated taste shocks to these baselines, so both levels drift with no pull back toward the baseline. 

Our simulated markets are populated rather than representative. A simulated consumer receives a taste multiplier for each variety and chooses the exact sparse demand \eqref{eq:sparse_share} at those tastes.\footnote{Each taste multiplier is drawn from a lognormal distribution with mean one and standard deviation \(0.30\) in logs. The unit mean design allows heterogeneity to add noise around the structural tastes without shifting them.} She then equalizes marginal rates of substitution at perceived prices while exhausting her true budget. The market share of a variety is the average of these individual shares across the consumers.\footnote{Exhausted budgets force every period's shares to sum to one. We check this construction period by period in our simulations.} Finally, we double-difference against the reference variety and assemble the regression variables required by the Feenstra method.\footnote{We make the panels balanced to keep every moment of the estimator traceable. The panel-growth requirement of Assumption \ref{ass:reg}(iii) is exercised directly by the \(T\) sweeps below.}

The first question we consider is how far the populated market sits from the representative consumer benchmark. Setting the dispersion to zero collapses every consumer onto the same sparse demand \eqref{eq:sparse_share}. This procedure produces the representative consumer of the model as a natural benchmark. Figure \ref{fig:het_conv} measures the distance in market shares on a fixed cross-section by aggregating between ten and five thousand consumers and averaging one hundred draws at each size. The maximum deviation across varieties falls with the square root of the number of consumers and then converges to a floor value.\footnote{This floor value is the aggregation gap between the average of a nonlinear share over taste draws and the share at the average taste, and it rises with attention because expenditure concentrates. At five thousand consumers, the deviation is at most \(0.005\) of expenditure at every attention weight, less than half of one percentage point of share.}

\begin{figure}[h]
\centering
\includegraphics[width=0.75\textwidth]{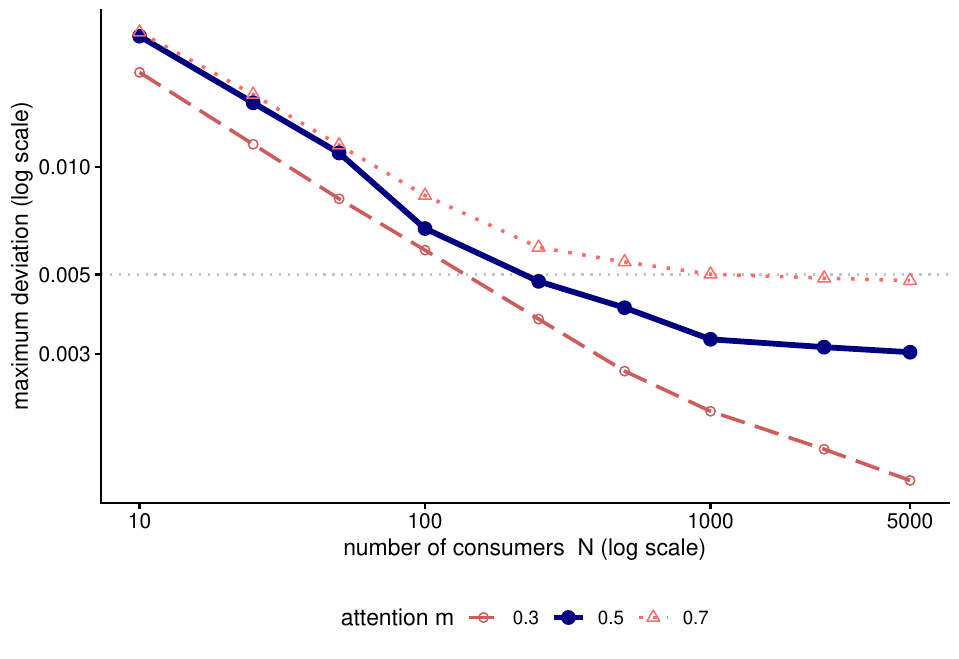}
\caption{Aggregation of heterogeneous shares to the representative agent. Maximum deviation between the market share of \(N\) consumers and the representative-agent share \eqref{eq:sparse_share}, on the baseline cross-section, averaged over one hundred draws, with log scales on both axes. The decline follows the square-root law until it meets the aggregation gap, which rises with attention as expenditure concentrates. At \(N=5{,}000\) the deviation is at most \(0.005\) at every attention weight.}
\label{fig:het_conv}
\end{figure}

The convergence figure (Figure\ref{fig:het_conv}) certifies the populated market shares in levels. The next relevant binding question is whether the exact demand equation survives aggregation in changes that are employed by the \citet{feenstra1994new} estimator. Figure \ref{fig:EH12} illustrates that on the long panel, it does.

\begin{figure}[h]
  \centering
  \includegraphics[width=0.98\textwidth]{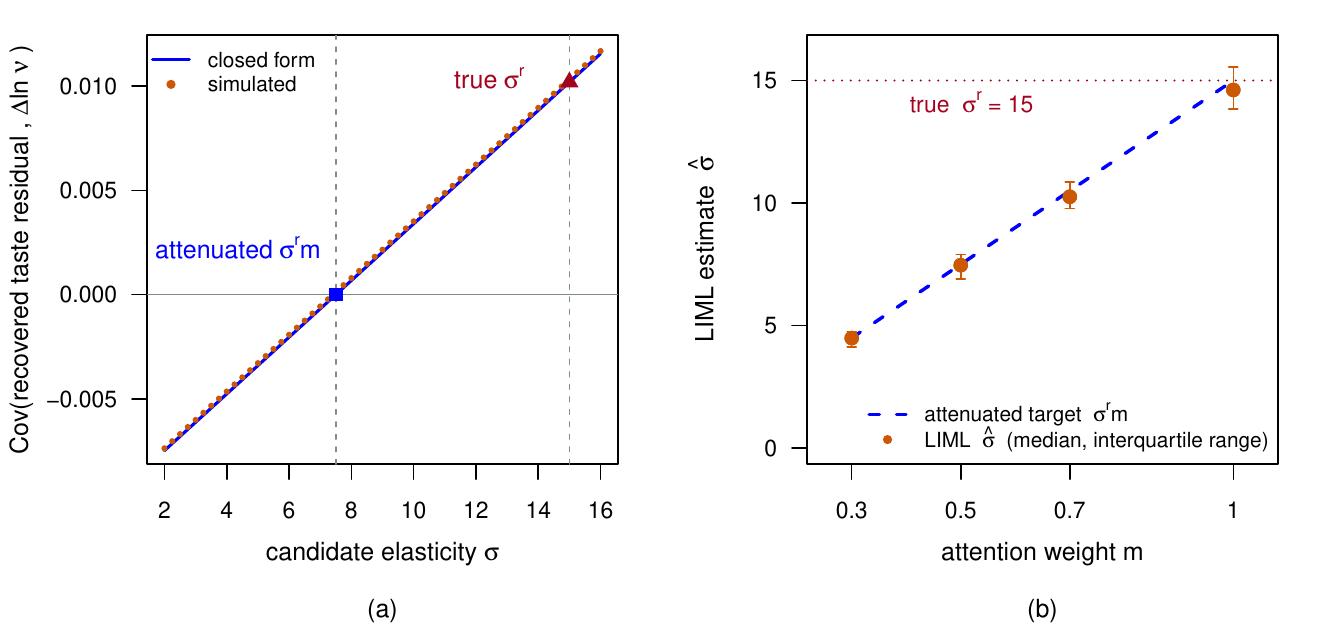}
  \caption{This pair of figures illustrates Proposition \ref{prop:main} using the simulated heterogeneous consumer market data. In Panel (a), the identifying moment is a function of the candidate elasticity at the focal attention weight \(m=0.5\). We plot simulated values (dots) and the closed form (line). The identifying moment is zero at the attenuated \(\sigma^r m = 7.5\) and positive at the true \(\sigma^r = 15\). In Panel (b), we plot the LIML estimates across the attention grid, displaying medians with interquartile ranges over one hundred fifty replications per attention weight at \(T = 120\). The median marks land on the attenuated line \(\sigma^r m\), reaching the true \(\sigma^r\) only at full attention (i.e., \(m=1\)).}
  \label{fig:EH12}
\end{figure}

\subsection{The Moment Fails at the Truth} 
Panel (a) of Figure \ref{fig:EH12} plots the identifying moment against the candidate elasticity on one long populated panel with six hundred periods of the market described above at the focal attention weight \(m = 0.5\). In this figure, each dot denotes a sample covariance at the corresponding candidate \(\sigma\) between the recovered residual \eqref{eq:recovered_resid} and the relative cost shock, computed by stacking the fourteen non-reference varieties over all periods. The line is the closed form \eqref{eq:moment}, derived for the representative consumer and evaluated at the panel's own cost-shock variance. 

In Figure \ref{fig:EH12} Panel (a), the dots never leave the line, the largest gap over the whole grid being one percent of the moment at the truth. So the line and the dots measure the same object. The moment crosses zero at the attenuated value \(\sigma^r m = 7.5\) and becomes positive at the true \(\sigma^r = 15\). A market of five thousand heterogeneous consumers, therefore, presents the estimator with the same orthogonality pattern one consumer implies, satisfied at the attenuated value but violated at the truth.  

Figure \ref{fig:EH12} Panel (a) is then the sample counterpart of Proposition \ref{prop:main}(i), computed on populated data. The violation at the truth exists because the residual the truth implies carries the term \(\sigma^r(1-m)\dln p_{vt}\) while prices respond to cost shocks through the market. Hence, the residual moves together with the same shock the moment requires it to ignore. This failure in the identifying moment sits in the joint distribution of prices and shares rather than in the construction of the estimator, and therefore procedures built from those observed prices and shares within the market are unlikely to fix it.

Panel (b) of Figure \ref{fig:EH12} takes Proposition \ref{prop:main}(ii) to the data, running the full estimator across the grid \(m \in \{0.3, 0.5, 0.7, 1\}\), with one hundred fifty markets per cell at \(T = 120\). At full attention, the estimator works, returning a median of \(14.6\) against the true \(15\).\footnote{The representative-agent benchmark returns \(14.67\), so the shortfall splits into \(0.33\) of the finite-sample bias of \citet{soderbery2015estimating} and an aggregation tilt of \(0.07\), the second-order wedge of the consumer sweep. The benchmarks at the remaining weights are \(4.48\), \(7.51\), and \(10.31\).} Away from full attention, the medians land at \(4.5\), \(7.5\), and \(10.3\) against targets of \(4.5\), \(7.5\), and \(10.5\), with the benchmarks trace mostly to the same finite-sample bias, while interquartile ranges never approach \(\sigma^r\). Panel (b) of Figure \ref{fig:EH12} therefore shows that a precise and well-behaved estimator of \(\sigma^r m\) lands on the attenuated line and it recovers \(\sigma^r\) only when attention is complete (i.e., \(m=1\)).

In Panel (b) of Figure \ref{fig:EH12}, the simulation runs for one hundred twenty periods. A practitioner could therefore attribute everything in that panel to weak instruments and short panels, which corresponds to the finite-sample bias interpretation of \citet{soderbery2015estimating}. One way to distinguish this objection from the failure of the identifying moment is to vary the amount of data and observe which concerns disappear. Figure \ref{fig:EH3} does that.

\begin{figure}[h]
  \centering
  \includegraphics[width=0.98\textwidth]{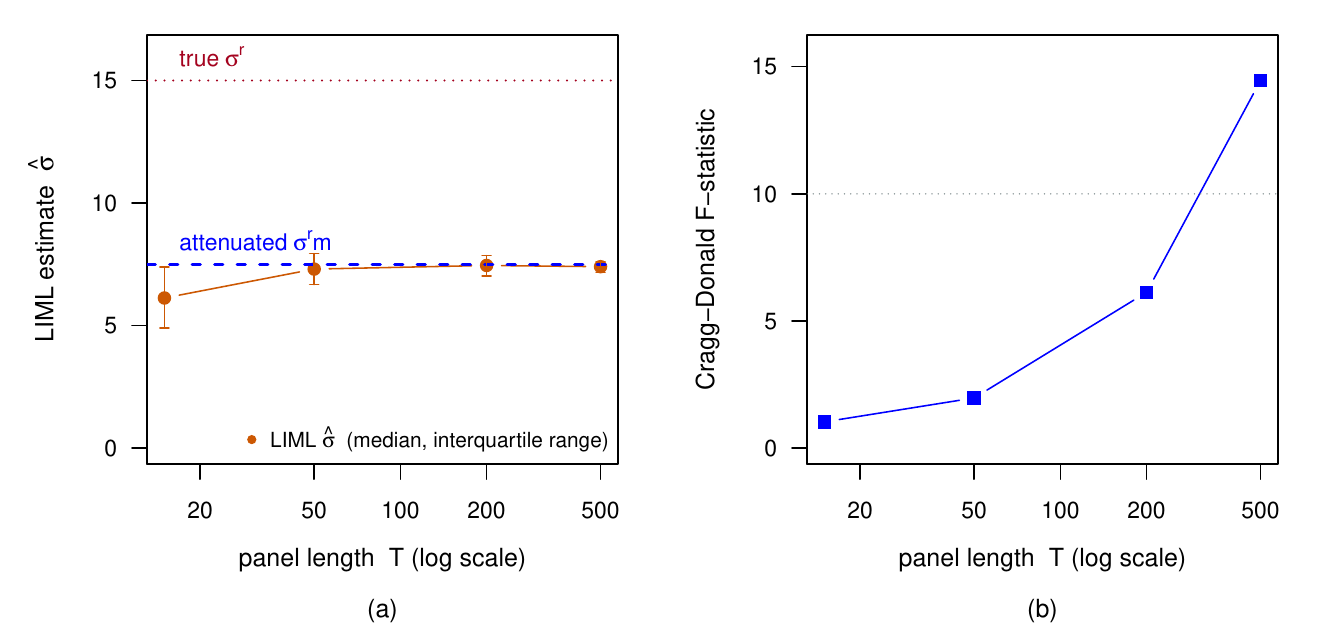}
  \caption{This pair of figures carry out a data-panel length sweep using the simulated heterogeneous consumer market data with the focal attention weight fixed at \(m = 0.5\) and the attenuated value being \(\sigma^r m = 7.5\). Panel (a) plots the LIML estimates, medians, and interquartile ranges across one hundred fifty replications per panel length, against the true \(\sigma^r = 15\) (dotted) and the attenuated \(\sigma^r m\) (dashed). The estimate settles at the attenuated value, and expanding the simulation panel narrows the interquartile range from \(2.5\) to \(0.45\) without moving the median toward the truth. Panel (b) plots medians of the joint Cragg-Donald F-statistic for the two endogenous regressors across the same set of replications. The median rises from \(1\) to \(14.5\) through the conventional rule-of-thumb threshold of ten (dotted). Longer panels sharpen the estimate of \(\sigma^r m\) and strengthen the instruments. But the distance to \(\sigma^r\) never shrinks.}
  \label{fig:EH3}
\end{figure}

\subsection{More Periods Sharpen the Wrong Number}

Figure \ref{fig:EH3} fixes the focal attention \(m=0.5\), so the attenuated value is \(\sigma^r m = 7.5\), and sweeps the simulated data panel length from fifteen periods to five hundred. Based on Assumption \ref{ass:reg}(iii), this sweep separates three explanations for a low estimate, too few periods, weak instruments, and identification, by curing the first two through expanding the data panel.

Panel (a) of Figure \ref{fig:EH3} shows that the median of the estimates sits at \(6.1\) at fifteen periods, reflecting the finite-sample bias that \citet{soderbery2015estimating} studies. This shortfall shrinks to a few percent by fifty periods and is gone by two hundred where the medians settling just below \(7.5\) because of the aggregation gap. Accordingly, the interquartile range narrows from \(2.5\) at the shortest panel to \(0.45\) at the longest while the median stays put. Therefore, a five hundred-period panel tightly identifies the attenuated estimand of \(7.5\), but leaves the true structural elasticity of \(15\) entirely unrecovered.

Panel (b) of Figure \ref{fig:EH3} removes the second explanation, weak instruments. The joint Cragg-Donald F-statistic for the two endogenous regressors rises from \(1\) to \(14.5\), through the conventional rule-of-thumb threshold of ten, as more periods resolve the cross-variety heteroskedasticity the variety instruments use.\footnote{The regression is Feenstra's product-moment equation across varieties, in which the squared change in the double-differenced price is fitted on the squared share change and on the price-share cross product, all relative to a reference variety \(k\). The variety indicators instrument these two endogenous regressors, and they are relevant only through the cross-variety spread in the taste-to-cost variance ratio.} As diagnostics improve with the panel length, the estimate never moves toward \(\sigma^r\) because the attenuation here is a property of the population moment, not of the sampling distribution.

Simulations conducted in this section confirm that the standard estimator converges to the attenuated value when consumers exhibit bounded attention. The gap between the returned pseudo-true parameter, in the sense of \citet{andrews2026true}, and the true structural elasticity persists in the population limit. Therefore, when not being properly accounted for, bounded attention results in an identification failure that finite-sample and weak-instrument corrections cannot resolve.

\section{Discussion and Conclusions}
\label{sec:conclusion}

Once attention is bounded, prices and quantities alone are insufficient to point-identify the structural elasticity. The standard supply-side approach yields a composite parameter that conflates preferences and attention. Structural elasticity is set-identified, with the composite serving as its lower bound, and the two parameters coincide only under full attention. This represents a population identification failure, which is distinct from the finite-sample and weak-instrument issues previously examined in the literature. Any combination of structural elasticity and attention weight that produces the same composite parameter fits the observed data equally well, although these combinations imply different welfare outcomes.

Therefore, recovering the structural elasticity requires information beyond price and quantity data. One approach involves specifying a model that rigorously characterizes how attention is allocated. Alternatively, an instrumental variable that shifts attention without affecting preferences can be employed, and recent policy interventions have begun to provide such instruments. 

For example, the French display rule requires supermarkets larger than 400 square meters to post the unit price change of a downsized product on the shelf for two months. Smaller stores and online retailers are exempt. If this notice alters shopper perception without changing the product, it satisfies the exclusion restriction required for an attention instrument. The exemptions allow the same product variety to be sold during the same period under two different salience regimes. Within the model, the display increases the attention weight for the treated product and store cells. In the limit of full attention, the price response in treated stores identifies the structural elasticity. The ratio of the response in exempt stores to that in treated stores identifies the attention weight. This decomposition separates the composite parameter into its constituent factors. This limit corresponds to the frictionless benchmark described by \citet{chetty2012bounds}, which is achieved through regulation. The estimation of these parameters is left for future research.

\appendix
\refstepcounter{section} 
\section*{Appendix. Proofs} 
\label{app:proofs}

Throughout, all changes are first differences relative to a reference variety \(k\) in the nest, so the common nest-time aggregates cancel and \(\eta_t\) drops out of the residuals.

\subsection{\texorpdfstring{Proof of Lemma \ref{lem:obs_equiv}}{Proof of Lemma 1}}

\begin{proof}
Start from the effective taste weight \(\tilde\phi_v = \phi_v (p_v^s/p_v)^{-\sigma^r} = \phi_v\, p_v^{\sigma^r}(p_v^s)^{-\sigma^r}\) in \eqref{eq:eff_taste} and insert it into the rational share \eqref{eq:rational_share} with \(\tilde\phi\) in place of \(\phi\),
\[
  \frac{\tilde\phi_v\, p_v^{1-\sigma^r}}{\sum_j \tilde\phi_j\, p_j^{1-\sigma^r}}
  = \frac{\phi_v\, p_v^{\sigma^r}(p_v^s)^{-\sigma^r}\, p_v^{1-\sigma^r}}{\sum_j \phi_j\, p_j^{\sigma^r}(p_j^s)^{-\sigma^r}\, p_j^{1-\sigma^r}}.
\]
In each term \(p_v^{\sigma^r}\cdot p_v^{1-\sigma^r} = p_v\), so the right-hand side equals \(p_v\phi_v(p_v^s)^{-\sigma^r}/W = s_v^s\), the sparse share \eqref{eq:sparse_share}. The bounded-attention shares thus coincide with rational shares under taste weights \(\tilde\phi_v\). An econometrician who fits \eqref{eq:rational_share} to \((s_v^s, p_v)\) recovers the taste weights that rationalize the shares, namely \(\tilde\phi_v\), and these differ from \(\phi_v\) by the factor \(R_v^{-\sigma^r}\) whenever \(m<1\) and \(p_v \neq \bar p\).
\end{proof}

\subsection{\texorpdfstring{Proof of the exact demand equation \eqref{eq:demand_exact}}{Proof of the exact demand equation}}

\begin{proof}
By Lemma \ref{lem:obs_equiv}, taking logs of the share gives \(\ln s_{vt}^s = \ln\tilde\phi_{vt} + (1-\sigma^r)\ln p_{vt} - \ln[\sum_j \tilde\phi_{jt} p_{jt}^{1-\sigma^r}]\).
By \eqref{eq:eff_taste}, \(\ln\tilde\phi_{vt} = \ln\phi_{vt} - \sigma^r \ln R_{vt}\).
The bracketed term depends only on the nest and the period, so first-differencing it and absorbing it into a period effect \(\eta_t\) gives \(\dln s_{vt}^s = \eta_t + \dln\phi_{vt} - \sigma^r\dln R_{vt} + (1-\sigma^r)\dln p_{vt}\).
From \eqref{eq:R}, \(\ln R_{vt} = (1-m)(\ln\bar p - \ln p_{vt})\), and the anchor is fixed, so \(\dln R_{vt} = -(1-m)\dln p_{vt}\).
If the anchor moves but is common across varieties, its change is part of the period effect and drops from relative changes.

Substituting,
\[
  \dln s_{vt}^s = \eta_t + \dln\phi_{vt} + \sigma^r(1-m)\dln p_{vt} + (1-\sigma^r)\dln p_{vt}
  = \eta_t + \dln\phi_{vt} + (1-\sigma^r m)\dln p_{vt}.
\]
No term is dropped and no expansion is taken, so the equation is exact, the coefficient on \(\dln p_{vt}\) is \(1-\sigma^r m\), and the residual is \(\dln\phi_{vt}\), which establishes \eqref{eq:apparent}.
\end{proof}

\subsection{\texorpdfstring{Proof of the reduced form \eqref{eq:reduced}}{Proof of the reduced form}}

\begin{proof}
In relative terms the pricing equation \eqref{eq:supply} gives \(\dln p_{vt} = \dln\nu_{vt} + \omega\dln x_{vt}\) with \(\dln x_{vt} = \dln s_{vt}^s - \dln p_{vt}\), so \(\dln p_{vt}(1+\omega) = \omega\dln s_{vt}^s + \dln\nu_{vt}\).
Substitute the demand equation \eqref{eq:demand_dbdiff}, written relative to \(k\), and collect the price terms,
\[
  \dln p_{vt}\big[1+\omega - \omega(1-\sigma^r m)\big] = \omega\,\dln\phi_{vt} + \dln\nu_{vt}.
\]
Since \(1+\omega-\omega(1-\sigma^r m) = 1+\omega\sigma^r m = D\), this is \eqref{eq:reduced}, with \(D>0\).
\end{proof}

\subsection{\texorpdfstring{Proof of the moment \eqref{eq:moment}}{Proof of the moment}}

\begin{proof}
Insert the demand equation \eqref{eq:demand_dbdiff} into the recovered residual \eqref{eq:recovered_resid}, relative to \(k\),
\[
  \dln\tilde\phi_{vt}(\sigma)
  = \dln\phi_{vt} + (1-\sigma^r m)\dln p_{vt} + (\sigma-1)\dln p_{vt}
  = \dln\phi_{vt} + (\sigma - \sigma^r m)\dln p_{vt}.
\]
Setting \(\sigma = \sigma^r\) gives
\[
  \dln\tilde\phi_{vt}(\sigma^r) \;=\; \dln\phi_{vt} + \sigma^r(1-m)\,\dln p_{vt} \;=\; \dln\tilde\phi_{vt},
\]
the change in the effective taste of Lemma \ref{lem:obs_equiv}, as claimed in the text.
Take the covariance with \(\dln\nu_{vt}\).
The term \(\Cov(\dln\phi_{vt}, \dln\nu_{vt})\) is zero by Assumption \ref{ass:orth}, and the reduced form \eqref{eq:reduced} gives
\[
  \Cov(\dln p_{vt}, \dln\nu_{vt}) \;=\; \frac{\Var(\dln\nu_{vt})}{D},
\]
again using Assumption \ref{ass:orth}.
Combining the two pieces gives 
\[\Cov(\dln\tilde\phi_{vt}(\sigma), \dln\nu_{vt}) = (\sigma - \sigma^r m)\Var(\dln\nu_{vt})/D.\]
At \(\sigma = \sigma^r m\) this is zero, and at \(\sigma = \sigma^r\) it equals \(\sigma^r(1-m)\Var(\dln\nu_{vt})/D\), which is strictly positive whenever \(m<1\).
\end{proof}

\subsection{\texorpdfstring{Proof of Proposition \ref{prop:main}}{Proof of Proposition 1}}

\begin{proof}

Consider a rational CES market with elasticity \(\sigma^\dagger = \sigma^r m\), the same inverse supply elasticity \(\omega\), taste shock \(\dln\phi_{vt}\), and cost shock \(\dln\nu_{vt}\), which are uncorrelated by Assumption \ref{ass:orth}.

Its demand equation is \(\dln s_{vt} = \eta_t + \dln\phi_{vt} + (1-\sigma^\dagger)\dln p_{vt}\) and its pricing equation is \eqref{eq:supply} with \(\sigma^\dagger\) in place of \(\sigma^r\).
Solving as in the reduced-form proof, its equilibrium price change is \((\omega\dln\phi_{vt} + \dln\nu_{vt})/(1+\omega\sigma^\dagger)\), and \(1+\omega\sigma^\dagger = 1+\omega\sigma^r m = D\). This price change is identical to the bounded-attention reduced form \eqref{eq:reduced}, and the rational share change \(\eta_t + \dln\phi_{vt} + (1-\sigma^r m)\dln p_{vt}\) is identical to the bounded-attention demand equation \eqref{eq:demand_exact}. The two economies therefore generate the same joint distribution of \((\dln p_{vt}, \dln s_{vt})\) period by period, with the same shocks, which proves the observational-equivalence claim and statement (i), as the moment \eqref{eq:moment} is then the moment of the rational economy with elasticity \(\sigma^r m\).

For statement (ii), the Feenstra and Broda-Weinstein estimator, in its product-moment and limited-information maximum-likelihood forms, is a measurable function of the joint distribution of \((\dln p_{vt}, \dln s_{vt})\) and the variety indicators. The constructed rational economy is correctly specified and satisfies Assumption \ref{ass:reg}(i) to (iii). The mapping from the slopes to the elasticity is well posed, since for \(\omega > 0\) Feenstra's quadratic in \(\sigma - 1\) has a positive leading coefficient and constant term \(-1\), so its two roots are real with opposite signs and exactly one satisfies \(\sigma > 1\). Assumption \ref{ass:reg}(iv) therefore delivers consistency for \((\sigma^\dagger, \omega) = (\sigma^r m, \omega)\). Because the bounded-attention data have the same joint distribution, the estimator applied to them has the same probability limit, \(\plim\hat{\sigma} = \sigma^r m\) and \(\plim\hat\omega = \omega\). The estimator cannot return \(\sigma^r\), because the data it sees are generated, in distribution, by a market whose elasticity is \(\sigma^r m\).
\end{proof}

%%%

\subsection{The additive perception rule \label{app:level}}

The exactness of the constant \(\sigma^r m\) relies on the log form of the perception rule \eqref{eq:perceived}, while the identification failure does not. We take \(\dln\) over time alone with the period effect explicit, the anchor stays fixed, and the pricing equation and Assumption \ref{ass:orth} carry over from the main text.

Let perception be additive in levels, the \citet{gabaix2014sparsity} original form,
\begin{equation}\label{eq:levels_rule}
  p_v^s \;=\; m\,p_v + (1-m)\,\bar p,
  \qquad
  R_{vt} \;=\; \frac{p^s_{vt}}{p_{vt}} \;=\; m + (1-m)\,\frac{\bar p}{p_{vt}}.
\end{equation}
Because the effective taste \(\tilde\phi_{vt} = \phi_{vt}\,R_{vt}^{-\sigma^r}\) came from the budget adjustment alone and never used the form of \(R_{vt}\), the observational equivalence of Lemma \ref{lem:obs_equiv} is unchanged. Taking logs of the equivalent share and differencing over time, with the nest aggregate absorbed into the period effect, gives the exact demand equation
\begin{equation}\label{eq:levels_exact}
  \dln s^{s}_{vt} \;=\; \eta_t + \dln\phi_{vt} + (1-\sigma^r)\,\dln p_{vt} - \sigma^r\,\dln R_{vt},
\end{equation}
in which the price enters twice, once directly and once through the nonlinear wedge.

%%%%
Differentiating the rule \eqref{eq:levels_rule} directly gives \(dp^s_{vt}/dp_{vt} = m\), so the level pass-through is the constant \(m\) everywhere. We take the following steps to show that the rule passes the varying fraction \(m/R_{vt}\) through in logs, which is what the demand system uses. 

First, rearranging the closed form of the wedge in \eqref{eq:levels_rule} gives the identity \((1-m)\,\bar p/p_{vt} = R_{vt} - m\), so \(R_{vt} > m\) and the ratio \(m/R_{vt}\) lies strictly between zero and one whenever \(m < 1\). Next, differentiating \(R_{vt} = m + (1-m)\,\bar p/p_{vt}\) in the price gives \(dR_{vt}/dp_{vt} = -(1-m)\,\bar p/p_{vt}^{2}\), and the chain rule \(d\ln x/d\ln p = (p/x)\,dx/dp\) converts this to
\begin{equation}\label{eq:levels_dlnR}
  \frac{d\ln R_{vt}}{d\ln p_{vt}} \;=\; -\,\frac{(1-m)\,\bar p/p_{vt}}{R_{vt}}.
\end{equation}
Substituting the rearrangement from \eqref{eq:levels_rule} into the numerator of \eqref{eq:levels_dlnR} yields \(d\ln R_{vt}/d\ln p_{vt} = -(R_{vt} - m)/R_{vt} = -(1 - m/R_{vt})\), which is negative whenever \(m < 1\), so the wedge falls as the price rises. In the meantime, because \(p^s_{vt} = R_{vt}\,p_{vt}\) by the definition of the wedge, \(\ln p^s_{vt} = \ln p_{vt} + \ln R_{vt}\), and differentiating gives \(d\ln p^s_{vt}/d\ln p_{vt} = 1 - (1 - m/R_{vt}) = m/R_{vt}\). Then, we substitute the closed form of \(R_{vt}\) to assemble the two derivatives into
\begin{equation}\label{eq:levels_elastic}
  \frac{d\ln R_{vt}}{d\ln p_{vt}} \;=\; -\Big(1 - \frac{m}{R_{vt}}\Big),
  \qquad
  \frac{d\ln p^s_{vt}}{d\ln p_{vt}} \;=\; \frac{m}{R_{vt}} \;=\; \frac{m\,p_{vt}}{m\,p_{vt} + (1-m)\,\bar p}.
\end{equation}

%%%%

Denote \(d\ln p^s_{vt}/d\ln p_{vt} = m/R_{vt}\) from \eqref{eq:levels_elastic} the log pass-through, the fraction of a log price change that enters the perceived log price. At the anchor, \(R_{vt} = 1\), the log pass-through is \(m\), and the demand slope below is \(1 - \sigma^r m\), so the additive and log rules coincide there. The coincidence is first order because around \(\bar p\) the second-order coefficient of the perceived price is \(m(m-1)/(2\bar p)\) under the log rule and zero under the additive rule. Therefore, the two data generating processes are second order apart for the small price changes that dominate high-frequency data. Away from the anchor, the log pass-through rises with the price, as its derivative in the log price is \((m/R_{vt})(1 - m/R_{vt}) > 0\).

The split around the anchor follows from two rearrangements of \eqref{eq:levels_rule}. In levels, \(p^s_{vt} - p_{vt} = (1-m)(\bar p - p_{vt})\) and \(p^s_{vt} - \bar p = m(p_{vt} - \bar p)\), so the perceived price always lies between the posted price and the anchor, and the misperception is the same fraction \(1-m\) of the gap on both sides. In ratios, denote \(1 - R_{vt} = (1-m)(1 - \bar p/p_{vt})\). Consider the sign of \(1 - R_{vt}\), whenever \(m < 1\), the wedge exceeds one if the price is below the anchor and falls short of one above it. Demand therefore looks least price sensitive where the discount is deepest, and the anchor deepens the misperception of discounts while it cushions the perception of increases.

Linearizing \(\dln R_{vt}\) through \eqref{eq:levels_elastic}, with the wedge treated as constant within each differencing window, gives the local demand equation
\begin{equation}\label{eq:levels_demand}
  \dln s^{s}_{vt} \;=\; \eta_t + \dln\phi_{vt} + \big(1 - \sigma^r m_{vt}\big)\,\dln p_{vt},
  \qquad
  m_{vt} \;\equiv\; \frac{m}{R_{vt}},
\end{equation}
the one approximation in this derivation, exact to first order at the anchor, where it returns the slope \(1 - \sigma^r m\) of the main text. The apparent elasticity \(\sigma^r m_{vt}\) varies with the price and lies strictly below \(\sigma^r\) whenever \(m < 1\).
 
The estimation follows in four steps. First, substituting \eqref{eq:levels_demand} into the recovered residual \eqref{eq:recovered_resid} at a candidate \(\sigma\), net of the period effect, the direct and inverted price terms combine into \(\dln\tilde\phi_{vt}(\sigma) = \dln\phi_{vt} + (\sigma - \sigma^r m_{vt})\,\dln p_{vt}\), so the friction survives only in the local coefficient. Second, take the covariance of this residual with the cost shock itself. The taste term drops by Assumption \ref{ass:orth}, while the local weight stays inside the expectation because it is a function of the window price level, so
\begin{equation}\label{eq:levels_moment}
  E\!\big[\dln\tilde\phi_{vt}(\sigma)\,\dln\nu_{vt}\big]
  \;=\; \sigma\,E\!\big[\dln p_{vt}\,\dln\nu_{vt}\big]
  \;-\; \sigma^r\,E\!\big[m_{vt}\,\dln p_{vt}\,\dln\nu_{vt}\big].
\end{equation}
Third, setting \eqref{eq:levels_moment} to zero yields
\begin{equation}\label{eq:levels_target}
  \sigma^{\circ} \;=\; \sigma^r\,
  \frac{E\!\big[m_{vt}\,\dln p_{vt}\,\dln\nu_{vt}\big]}{E\!\big[\dln p_{vt}\,\dln\nu_{vt}\big]}
  \;\equiv\; \sigma^r\,E_w\!\Big[\frac{m}{R_{vt}}\Big],
\end{equation}
where \(E_w\) is the expectation that tilts each observation by its price and cost-shock product. Fourth, we establish that the tilting weights are nonnegative once the shocks are integrated out, though they can take either sign in a single realization. Conditional on the window state, the nest-time aggregates and the base price that fixes \(m_{vt}\), the shocks stay mean zero and mutually orthogonal, so the reduced form \eqref{eq:reduced} holds at the local weight and the conditional covariance of the price change with the cost shock is \(\Var(\dln\nu_{vt})/D_{vt}\) with \(D_{vt} = 1 + \omega\,\sigma^r m_{vt}\). By the law of iterated expectations, both expectations in \eqref{eq:levels_target} average this conditional covariance, which delivers
\begin{equation}\label{eq:levels_convex}
  \sigma^{\circ} \;=\; \sigma^r\,\frac{E\big[m_{vt}\,\lambda_{vt}\big]}{E\big[\lambda_{vt}\big]},
  \qquad
  \lambda_{vt} \;\equiv\; \frac{\Var(\dln\nu_{vt})}{D_{vt}} \;>\; 0,
\end{equation}
a convex combination of the local values with weights \(\lambda_{vt}/E[\lambda_{vt}]\) that are nonnegative and average to one. Viewed as a function of the candidate \(\sigma\), the slope of the moment \eqref{eq:levels_moment} equals \(E[\lambda_{vt}] > 0\), so the zero is unique. The gap below the structural elasticity is \(\sigma^r - \sigma^{\circ} = \sigma^r\,E_w[1 - m_{vt}]\), strictly positive whenever \(m < 1\). The target collapses to \(\sigma^r m\) when every price sits at the anchor and to \(\sigma^r\) under full attention. 

The estimator observes no cost shock and weights the local equations nonlinearly, so its limit is a different nonnegative-weight average of the same local values. Every such average inherits the bound, so the estimator returns a value strictly below the structural elasticity, a weighted average of the local attenuations rather than the single constant \(\sigma^r m\).

%%%

\bibliography{behavioralCES}

\end{document}